%% file: example.tex
\documentclass[submission,copyright,creativecommons]{eptcs}
\providecommand{\event}{ICLP 2025} %

\usepackage{iftex}

\ifpdf
  \usepackage{underscore}         %
  \usepackage[T1]{fontenc}        %
\else
  \usepackage{breakurl}           %
\fi

\title{A Simple and Effective ASP-Based Tool for Enumerating Minimal Hitting Sets}
\author{Mohimenul Kabir
\institute{National University of Singapore}
\institute{School of Computing\\
Singapore}
\and
Kuldeep S Meel
\institute{Georgia Institute of Technology}
\institute{School of Computer Science\\
USA}
}
\def\titlerunning{An ASP-Based Tool for Enumerating MHSes}
\def\authorrunning{Kabir and Meel}
\input{section/macro}
\begin{document}
\maketitle

\begin{abstract}
  The hitting set problem is a fundamental problem in computer science and mathematics. Given a family of sets over a universe of elements, a minimal hitting set is a subset-minimal collection of elements that intersects each set in the family. Enumerating all minimal hitting sets is crucial in various real-world applications. 

  In this paper, we address the full enumeration of all minimal hitting sets for a given family of sets. 
  We formulate the problem using Answer Set Programming (ASP) and leverage existing ASP solvers for efficient enumeration. 
  We propose an ASP-based tool, \toolname, and our empirical evaluation shows that it effectively enumerates minimal hitting sets across benchmarks from diverse problem domains.  
\end{abstract}


\input{section/introduction}

\input{section/preliminaries}

\input{section/methods}

\input{section/experiment}

\input{section/conclusion.tex}

\nocite{*}
\bibliographystyle{eptcs}
\bibliography{example}
\end{document}

%% file: section/macro.tex
\usepackage{amsmath}
\usepackage{amsthm}
\usepackage{amssymb} %
\usepackage{array} %
\usepackage{amsfonts}
\usepackage{cleveref}
\usepackage{flexisym}
\usepackage{xfrac}
\usepackage{algorithm}
\usepackage{algpseudocode}
\usepackage{bbm}
\usepackage{xspace}
\usepackage{tikz}
\usepackage{booktabs}
\usepackage{subcaption}
\usepackage[mathcal]{eucal}
\usepackage{graphicx}
\DeclareGraphicsExtensions{.pdf,.png,.jpg}
\usetikzlibrary{positioning}

\newcommand{\Card}[1]{|#1|}

\newcommand{\true}{\ensuremath{\mathsf{true}}\xspace}
\newcommand{\false}{\ensuremath{\mathsf{false}}\xspace}

\newcommand{\toolname}{\ensuremath{\mathsf{MinHit}}-\ensuremath{\mathsf{ASP}}}

\algnewcommand\algorithmicforeach{\textbf{for each}}
\algdef{S}[FOR]{ForEach}[1]{\algorithmicforeach\ #1\ \algorithmicdo}

\newcommand{\rules}[1]{#1}
\newcommand{\dis}{\mathsf{dis}}
\newcommand{\body}[1]{\mathsf{Body}(#1)}
\newcommand{\at}[1]{\mathsf{atoms}(#1)}
\newcommand{\head}[1]{\mathsf{Head}(#1)}
\newcommand{\answer}[1]{\mathsf{AS}(#1)}
\newcommand{\dlp}[1]{\mathsf{HS2ASP}(#1)}
\newcommand{\set}{\mathcal{S}}
\newcommand{\universe}{\mathcal{U}}
\newcommand{\itemtovar}[1]{g(#1)}
\newcommand{\hittingset}[1]{\mathsf{MHS}(#1)}

\usepackage{listings}
\renewcommand{\lstlistingname}{Encoding}
\usepackage{xcolor}
\definecolor{codegreen}{rgb}{0,0.6,0}
\definecolor{codegray}{rgb}{0.5,0.5,0.5}
\definecolor{codepurple}{rgb}{0.58,0,0.82}
\definecolor{backcolour}{rgb}{0.95,0.95,0.92}
\newtheorem{lemma}{Lemma}

\newtheorem{example}{Example}

%% file: section/introduction.tex
\section{Introduction}
Given a family of sets $\set = \{S_1, \ldots, S_k\}$, a {\em hitting set} of $\set$ is a set $h$ that intersects every set of $\set$~\cite{VO2000}.
A hitting set $h$ is called {\em minimal hitting set} (MHS) if no proper subset of $h$ is a hitting set of $\set$.
In the MHS problem, each $S_i \in \set$ is a {\em set} and the union $\universe = \bigcup_{i} S_i$ is called the {\em universe}.
The problem definition is closely related to {\em minimal hypergraph traversal}~\cite{GMKT1997}.
The minimal hitting sets have applications in a wide range of domains, such as diagnosis~\cite{AG2009,BS2005,Reiter1987},
explanability~\cite{ILSM2021,SU2006}, data mining~\cite{BMR2003,BGKM2003,GMKT1997}, bioinformatics and computational biology~\cite{KG2004,TWS2009}, and combinatorial games~\cite{MTC2014}.

In many cases, a family of sets $\set$ has multiple minimal hitting sets (MHSes); in fact, the number of MHSes can be exponential~\cite{CKMV2011}. In such settings, enumerating all MHSes is often desirable to gain a deeper understanding of the underlying problem. 
For instance, consider a faulty system where each {\em diagnosis} corresponds to a hitting set of the problem~\cite{Reiter1987}. Enumerating all hitting sets provides a complete characterization of the possible fault diagnoses. Similarly, the OCSANA framework~\cite{VBBZ2013} computes all MHSes of {\em elementary paths} to analyze {\em signal transduction networks}.

A wide range of techniques have been developed for computing minimal hitting sets (MHSes). To the best of our knowledge, the first attempt to compute MHSes was introduced by Berge~\cite{Berge1984} in the context of hypergraph traversal. Later, Reiter proposed a method for computing MHSes using a specialized data structure known as the {\em hitting set tree}~\cite{Reiter1987}.  
Over the years, numerous MHS algorithms have been developed across various domains, including hypergraph traversal~\cite{Berge1984,MU2013}, model-based diagnosis~\cite{GSW1989}, data mining~\cite{BMR2003,DL2005}, poset theory~\cite{LMLX2010}, Boolean algebra~\cite{FK1996}, integer programming~\cite{BEGKM2002}, binary decision diagrams~\cite{Toda2013}, and evolutionary computation~\cite{LY2002}.  
In contrast, we adopt an alternative approach by formulating MHS enumeration as an {\em answer set program}~\cite{GKS2012}.

Answer Set Programming (ASP)~\cite{MT1999} is a well-established declarative paradigm in knowledge representation and reasoning, widely used for modeling complex combinatorial problems~\cite{EGL2016}. 
An ASP {\em program} consists of \textit{logical rules (grounded)} over {\em propositional atoms}, encoding domain knowledge and queries. 
A solution to an ASP program, known as an \emph{answer set}, is an assignment of values to propositional atoms that satisfies the program’s semantics.  
ASP programs are classified into two categories based on their expressiveness: {\em disjunctive} and {\em normal} ASP programs~\cite{Lierler2005}. However, the higher expressiveness comes at the cost of computational complexity. Determining whether a disjunctive ASP program has an answer set is $\textstyle \sum_{2}^{P}$-complete~\cite{EG1995}, whereas for normal ASP programs, the complexity reduces to NP-complete~\cite{MT1991}.

The main contribution of this paper is the design, implementation, and experimental evaluation of a new MHS enumeration tool called \toolname. 
\toolname~exploits the efficiency of ASP solvers for the MHS enumeration. 
More specifically, \toolname~reduces MHS enumeration to answer set enumeration problem of a {\em head-cycle free} disjunctive answer set program~\cite{Truszczynski2011}.
Note that the MHS problem has also been reduced to ASP solving in earlier work~\cite{AD2016}.
According to the computational complexity, the minimal hitting set problem is NP-hard~\cite{Karp2009,RNW1983}.
The decision problem for disjunctive logic programs lies in $\textstyle \sum_{2}^{P}$~\cite{EG1995}.
It is worth noting that reducing one problem to a harder one is not uncommon in combinatorial problem-solving~\cite{LSK2023}. 
It is well known to exploit ASP solvers in real-world diagnosis problems~\cite{GVETCSS2013,WK2022} and related problems~\cite{KM2024,KM2025}.
To the end, we conduct experimental evaluation to evaluate \toolname~on real-world hitting set benchmarks.
Our extensive evaluation reveals that \toolname~outperforms existing systems of all MHSes enumeration problem.

The remainder of this paper is organized as follows: \Cref{section:preliminaries} provides the necessary background to understand our main contribution. \Cref{section:methods} describes the methodology and implementation of \toolname. \Cref{section:experiment} presents the experimental evaluation of \toolname\ on real-world benchmark datasets. Finally, we conclude our study in \Cref{section:conclusion}.

%% file: section/preliminaries.tex
\section{Preliminaries}
\label{section:preliminaries}

\paragraph{Answer Set Programming.}
An \textit{answer set program} $P$ consists of a set of rules, each rule is structured as follows:
\begin{align}
\label{eq:generalrule}
\text{Rule $r$:~~}a_1 \vee \ldots \vee a_k \leftarrow b_1, \ldots, b_m, \textsf{not } c_1, \ldots, \textsf{not } c_n
\end{align}
where, $a_1, \ldots, a_k, b_1, \ldots, b_m, c_1, \ldots, c_n$ are {\em propositional variables} or {\em atoms}, and $k,m,n$ are non-negative integers. 
The notations $\rules{P}$ and $\at{P}$ denote the rules and atoms within the program $P$. 
In rule $r$, the operator ``\textsf{not}'' denotes \textit{default negation}~\cite{clark1978}. For each 
rule $r$ (\cref{eq:generalrule}), we adopt the following notations: the atom set $\{a_1, \ldots, a_k\}$ constitutes the {\em head} of $r$, denoted by $\head{r}$, the set $\{b_1, \ldots, b_m\}$ is referred to as the {\em positive body atoms} of $r$, denoted by $\body{r}^+$, and the set $\{c_1, \ldots, c_n\}$ is referred to as the \textit{negative body atoms} of $r$, denoted by $\body{r}^-$.
We often represent empty body ($\body{r}^+ = \emptyset$ and $\body{r}^- = \emptyset$) by the notation $\top$.
A program $P$ is called a {\em disjunctive logic program} if $\exists r \in \rules{P}$ such that $\Card{\head{r}} \geq 2$~\cite{BD1994}.
In this paper, we use the ASP-Core-$2$ standard input language for ASP programs~\cite{CFGIKK2020}.

In ASP, an interpretation $M$ over $\at{P}$ specifies which atoms are assigned \true; that is, an atom $a$ is \true under $M$ if and only if $a \in M$ (or \false when $a \not\in M$ resp.). 
An interpretation $M$ satisfies a rule $r$, denoted by $M \models r$, if and only if $(\head{r} \cup \body{r}^{-}) \cap M \neq \emptyset$ or $\body{r}^{+} \setminus M \neq \emptyset$. An interpretation $M$ is a {\em model} of $P$, denoted by $M \models P$, when $\forall_{r \in \rules{P}} M \models r$. 
The \textit{Gelfond-Lifschitz (GL) reduct} of a program $P$, with respect to an interpretation $M$, is defined as $P^M = \{\head{r} \leftarrow \body{r}^+| r \in \rules{P}, \body{r}^- \cap M = \emptyset\}$~\cite{GL1991}.
An interpretation $M$ is an {\em answer set} of $P$ if $M \models P$ and no $M\textprime \subsetneq M$ exists such that $M\textprime \models P^M$.
We denote the answer sets of program $P$ using the notation $\answer{P}$.

%% file: section/methods.tex
\section{Methodology and Architecture}
\label{section:methods}
In this section, we present the methodology and architecture of \toolname. At a high level, \toolname\ transforms the MHS enumeration problem into an ASP solving task. We first describe the reduction process from MHS enumeration to ASP solving. Then, we outline the architecture of \toolname.

\paragraph{From MHS enumeration to ASP Solving.}
Given a family of sets $\set$, we define an operator $\dlp{\set}$ that constructs an ASP program. In the reduction to ASP solving, given the {universe} $\universe = \bigcup_{i} S_i$, we introduce a mapping $\itemtovar{\cdot}$, which maps each element $x \in \universe$ to a unique propositional atom $\itemtovar{x}$. 
The operator $\dlp{\set}$ is applied to each set $S_i = \{a_1, \ldots, a_{\ell}\} \in \set$, generating a corresponding rule $\itemtovar{a_1} \vee \ldots \vee \itemtovar{a_{\ell}} \leftarrow \top.$, as illustrated in Encoding~\ref{code:mhstoasp}. 
The intuition behind introducing $\dlp{\set}$ is that each model of $\dlp{\set}$ corresponds to a hitting set of $\set$, and vice versa. 
Moreover, each answer set of $\dlp{\set}$ corresponds precisely to a minimal hitting set of $\set$.
We formally prove that the answer sets of $\dlp{\set}$ correspond one-to-one with the minimal hitting sets of $\set$.

\begin{figure*}
    \begin{lstlisting}[caption={The operator $\dlp{\set}$ for a family of sets $\set$. The lines starting with \% sign are comments. We often use the notation $\top$ to denote empty rule body.},label={code:mhstoasp},captionpos=b,mathescape=true,escapechar=|,xleftmargin=0.1cm]
      $\itemtovar{a_1} \vee \ldots \vee \itemtovar{a_{\ell}} \leftarrow \top.$|\label{line:eachvariable}|
      \end{lstlisting}    
      \renewcommand{\lstlistingname}{Algorithm}
\end{figure*}

\begin{lemma}
    \label{lemma:mhstoaspproof}
    For a family of sets $\set$,
    \begin{enumerate}
        \item each minimal hitting set of $\set$ corresponds to an answer set of $\dlp{\set}$.
        \item each answer set of $\dlp{\set}$ corresponds to a minimal hitting set of $\set$.
    \end{enumerate}
\end{lemma}
\begin{proof}
    \textbf{Part of `1'}: Without loss of generality, assume that $h = \{a_1, \ldots, a_{\ell}\}$ is a minimal hitting set of $\set$. 
    We show that $M = \{\itemtovar{a_1}, \ldots, \itemtovar{a_{\ell}}\}$ is an answer set of $\dlp{\set}$.
    
    Note that $M \models \dlp{\set}$ and $\dlp{\set}^{M} = \dlp{\set}$, since there is no negation (\textsf{not}) in $\dlp{\set}$. 
    Clearly, there is no $M\textprime \subsetneq M$ such that $M\textprime \models \dlp{\set}^{M}$.
    Otherwise, $h$ is not a minimal hitting set. Thus, $M$ is an answer set of $\dlp{\set}$.

    \textbf{Part of `2'}: 
    Assume that there is an answer set $M = \{\itemtovar{a_1}, \ldots, \itemtovar{a_{\ell}}\} \in \answer{\dlp{\set}}$. We show that $h = \{a_1, \ldots, a_{\ell}\}$ is a minimal hitting set of $\set$.

    According to the answer set semantics, $\forall S_i \in \set$, $S_i \cap h \neq \emptyset$, i.e., $h$ is a hitting set of $\set$. 
    Clearly, there is no $h\textprime \subsetneq h$; otherwise, from the hitting set $h\textprime$, $\exists M\textprime \subset M$ such that $M\textprime \models \dlp{\set}^{M}$, thus $M \not \in \answer{\dlp{\set}}$ --- a contradiction. 
    So, $h$ is a minimal hitting set of $\set$.
\end{proof}

\paragraph{Architecture of \toolname.}
The architecture of \toolname\ is illustrated in~\Cref{figure:architecture}. Given a family of sets $\set$, \toolname\ constructs the corresponding ASP program $\dlp{\set}$, as defined in Encoding~\ref{code:mhstoasp}. It then invokes an ASP solver to enumerate the answer sets of $\dlp{\set}$. 

Notably, although the disjunctive ASP program $\dlp{\set}$ can be translated into a normal logic program~\cite{LL2003}, this translation results in a quadratic blowup~\cite{AD2016}.

\begin{figure}
    \centering
    \begin{tikzpicture}[node distance=1cm]
        \node (block1) [draw, rectangle, fill=blue!10, minimum width=2.5cm, minimum height=1.2cm] {$\dlp{\set}$};
        \node (block2) [draw, rectangle, fill=orange!20, right=of block1, minimum width=2.5cm, minimum height=1.2cm] {ASP Solver};
        
        \draw[->, thick] ([xshift=-1cm]block1.west) -- node[pos=0, anchor=east]{$\set$} (block1.west);

        \draw[->, thick] (block1.east) -- node[above]{} (block2.west);

        \draw[->, thick] (block2.east) -- node[pos=1, anchor=west]{$\hittingset{\set}$} ([xshift=1cm]block2.east);
    
    \end{tikzpicture}
    \caption{The high-level methodology of~\toolname.}
    \label{figure:architecture}
\end{figure}
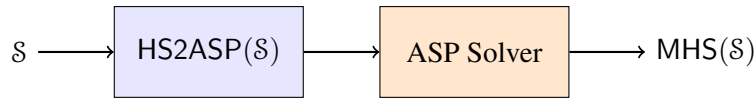

\begin{example}
    Consider the family of sets $\set = \{\{1,2\}, \{3\}, \{2,3,4\}\}$. 
    The sets $\set$ has two minimal hitting sets: $\{1,3\}$ and $\{2,3\}$.
    
    The ASP program $\dlp{\set}$ consists of rules: 
    $\{\itemtovar{1} \vee \itemtovar{2} \leftarrow \top.\text{ }\itemtovar{3} \leftarrow \top.\text{ }\itemtovar{2} \vee \itemtovar{3} \vee \itemtovar{4} \leftarrow \top.\}$, and 
    it has two answer sets: $\{\itemtovar{1}, \itemtovar{3}\}$ and $\{\itemtovar{2}, \itemtovar{3}\}$, which correspond to two MHSes of $\set$.
\end{example}

%% file: section/experiment.tex
\section{Experimental Evaluation}
\label{section:experiment}
We present the experimental evaluation of \toolname~on publicly available benchmarks from various domains. The primary objective is to assess how well \toolname\ performs compared to existing tools for minimal hitting set enumeration.
The dataset and code used in our experimental evaluation are available at: \url{https://zenodo.org/records/15108391}.

\paragraph{Design Choice}
We employed Clingo v$5.7.1$~\cite{GKS2009,GKS2013} as the underlying ASP solver in \toolname. We also evaluated Wasp v$2$~\cite{ADLR2015} (with the parameter \texttt{--disjunction=auto}), but we observed that Wasp did not outperform Clingo in our experiments\footnote{We observed that Clingo and Wasp are able to solve $1948$ and $1806$ instances, respectively.}. 

\paragraph{Benchmark and Baseline}
Our benchmark suite consists of instances drawn from real-world domains involving minimal hitting set problems, including unique column combinations~\cite{BBFNPS2020}, cluster vertex deletion~\cite{VS2020}, metabolic reactions~\cite{HK2011}, cell signaling networks~\cite{ZS2005}, the Connect-4 board game~\cite{DG2017}, frequent itemset mining~\cite{MU2013}, graph theory~\cite{BEGK2003,KS2005}, and combinatorial circuits from the ISCAS$85$ suite~\cite{KNPetal2009}. 
We also include randomly generated instances from~\cite{MU2013}. The benchmark set is primarily compiled from datasets used in~\cite{BFSW2022,MU2013}. These benchmarks contain up to $3,\!511$ elements and $1,\!973,\!734$ sets.

We compare \toolname\ against several existing systems capable of enumerating minimal hitting sets. Specifically, we evaluate the performance of Hitman~\cite{AAJ2018}, Sparsity-based Hypergraph Dualization (SHD)~\cite{MU2013}, and MtMiner~\cite{HBC2007}. For Hitman, we use the MCS-based technique~\cite{AAJ2018,MPM2015} (this setting enumerates MHSes in an unordered way, which makes for the fairest comparison.). 
We also include SAT-based~\cite{ES2003,LYZR2021} and Integer Linear Programming (ILP)-based approaches in our comparison. These techniques iteratively invoke a SAT solver and the Gurobi ILP solver (v$11.0.1$), respectively, to enumerate all minimal hitting sets.
Since MHSes correspond to {\em minimal models}, we do not compare against tools designed specifically for full AllSAT enumeration, such as HALL~\cite{FNS2023} and the method from~\cite{MSS2023}.

\paragraph{Environmental Settings} 
All experiments were conducted on a high-performance computing cluster, with each node consisting of Intel Xeon Gold $6248$ CPUs.
Each benchmark instance was executed on a single core, with a time limit of $1000$ seconds and a memory limit of $16$\,GB for all considered tools.

\paragraph{Experimental Results}
\begin{figure}
    \centering
    \includegraphics[width=0.7\linewidth]{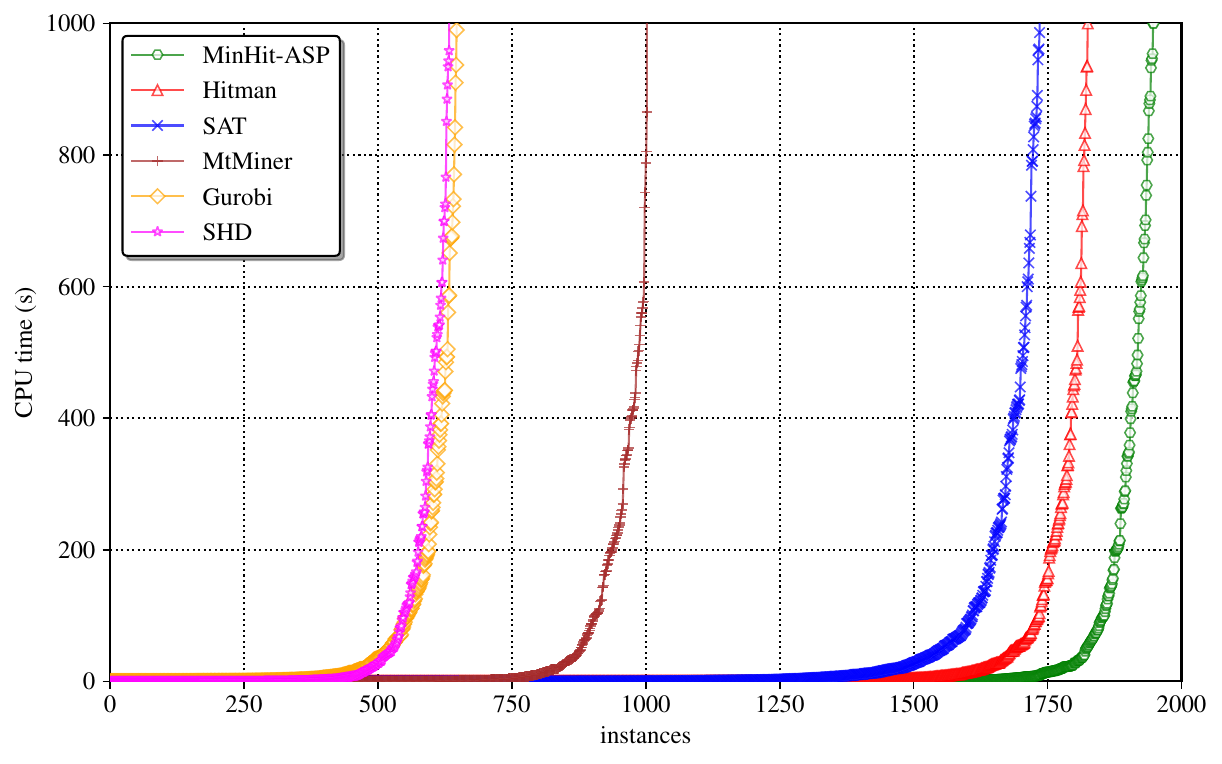}
    \caption{The runtime comparison of different minimal hitting set enumerators.}
\label{fig:runtime}
\end{figure}

\begin{table}[h]
    \centering
    \begin{tabular}{m{4em} m{3em} m{3em} m{4em} m{3em} m{4em} m{6em}} 
    \toprule
    & SAT & ILP & MtMiner & SHD & Hitman & \toolname\\
    \midrule
    \#Solved & 1735 & 647 & 1002 & 633 & 1846 & \textbf{1948}\\
    \bottomrule
    \end{tabular}
    \caption{The performance comparison of \toolname~vis-a-vis existing hitting set solver.}
    \label{table:hittingsetresult}
\end{table}
\begin{table}[h]
    \centering
    \begin{tabular}{m{4em} m{3em} m{3em}} 
    \toprule
    & VBS1 & VBS2 \\
    \midrule
    \#Solved & 1906 & 2000 \\
    \bottomrule
    \end{tabular}
    \caption{The performance comparison of virtual best solvers. The virtual best solver VBS1 combines all baselines except \toolname~and VBS2 combines all baselines including \toolname.}
    \label{table:virtualbestresult}
\end{table}

\begin{figure*}
    \centering
    \begin{subfigure}[t]{0.48\textwidth}
        \centering
        \includegraphics[width=0.9\linewidth]{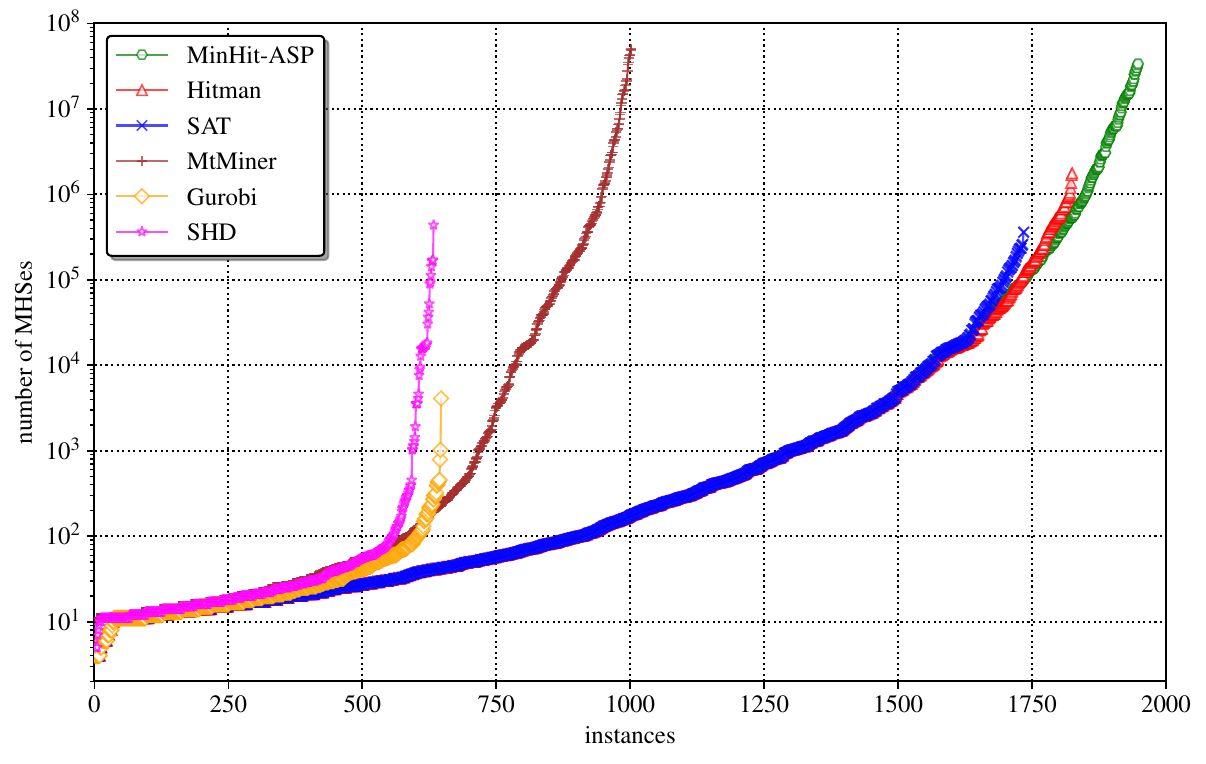}
        \caption{The performance comparison in terms of the number of minimal hitting sets enumerated.}
    \label{fig:countminimalhittingset}
    \end{subfigure}
    \begin{subfigure}[t]{0.48\textwidth}
        \centering
        \includegraphics[width=0.9\linewidth]{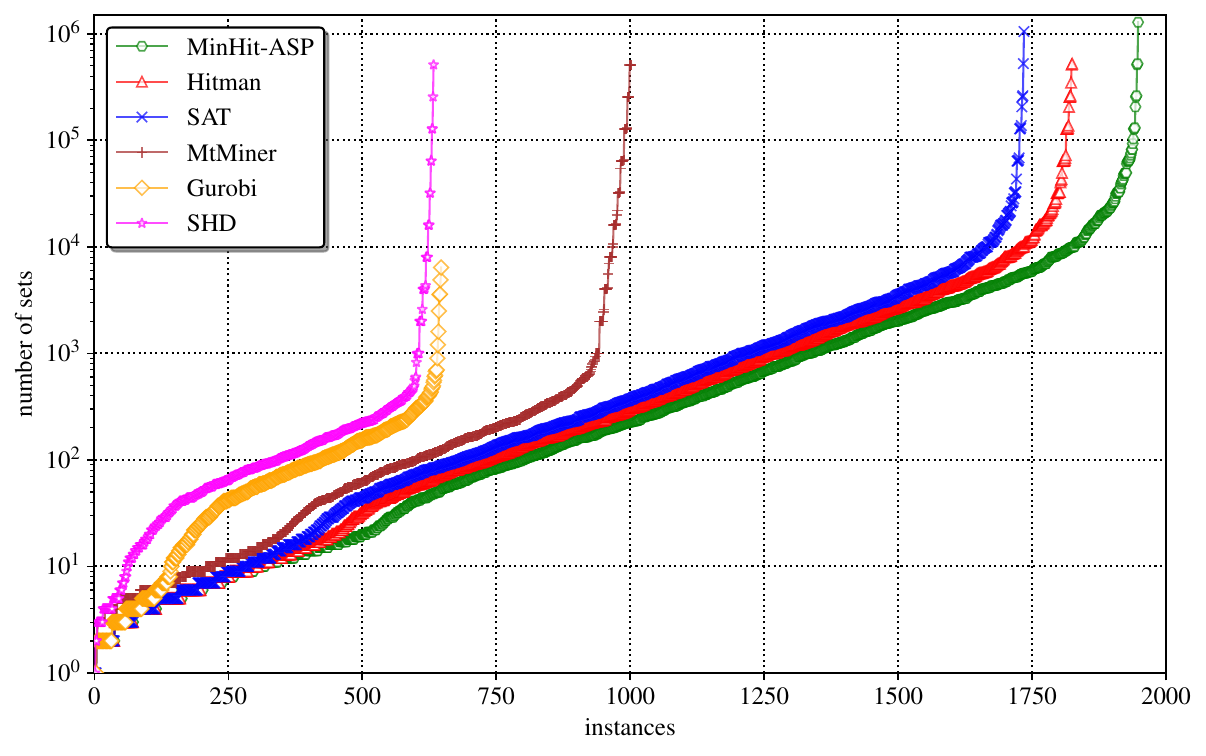}
        \caption{The performance comparison in terms of the size of instances.}
    \label{fig:sizeminimalhittingset}
    \end{subfigure}
    \begin{subfigure}[t]{0.48\textwidth}
        \centering
        \includegraphics[width=0.9\linewidth]{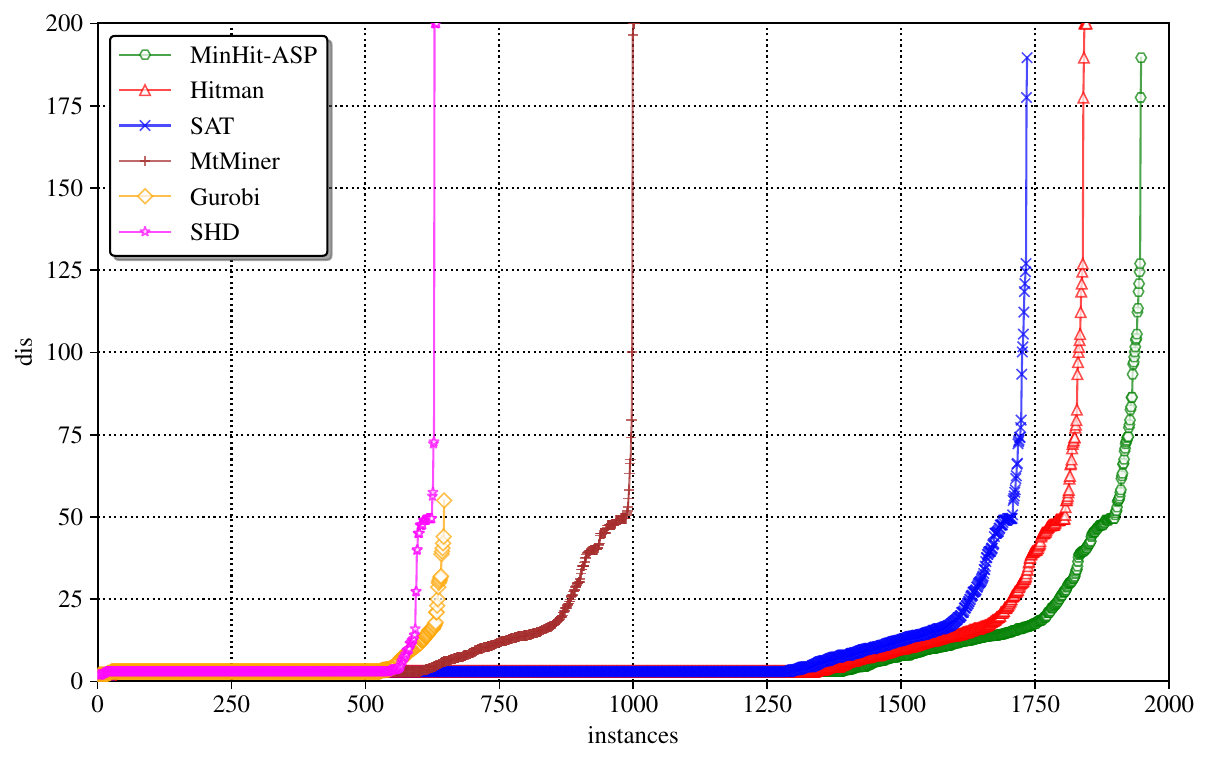}
        \caption{The performance comparison in terms of the average disjunction size.}
    \label{fig:disjunctionminimalhittingset}
    \end{subfigure}
    \caption{The performance comparison of \toolname~vis-a-vis existing minimal hitting set solvers.}
    \label{fig:comparison}
\end{figure*}

\begin{table}[h]
    \centering
    \begin{tabular}{m{7em} m{3em} m{3em} m{3em} m{3em} m{3em} m{3em} m{6em}} 
    \toprule
    & $\sum$ & SAT & ILP & MtMiner & SHD & Hitman & \toolname\\
    \midrule
    $1 \leq \Card{\set} <10^1$ & 3 & 3 & 3 & 0 & 1 & 3 & 3\\
    \midrule
    $10^1 \leq \Card{\set} <10^2$ & 320 & 285 & 141 & 225 & 61 & 302 & \textbf{315}\\
    \midrule
    $10^2 \leq \Card{\set} <10^3$ & 1312 & 397 & 280 & 376 & 280 & 426 & \textbf{483}\\
    \midrule
    $10^3 \leq \Card{\set} <10^4$ & 595 & 529 & 216 & 341 & 264 & 532 & \textbf{533}\\
    \midrule
    $10^4 \leq \Card{\set} <10^5$ & 524 & 452 & 7 & 25 & 16 & 478 & \textbf{494}\\
    \midrule
    $10^5 \leq \Card{\set} <10^6$ & 187 & 60 & 0 & 21 & 7 & 93 & \textbf{110}\\
    \midrule
    $10^6 \leq \Card{\set} <10^7$ & 69 & 8 & 0 & 14 & 4 & 9 & 9\\
    \midrule
    $10^7 \leq \Card{\set} <10^8$ & 4 & 1 & 0 & 0 & 0 & \textbf{3} & 1\\
    \bottomrule
    \end{tabular}
    \caption{The performance comparison of \toolname~vis-a-vis existing hitting set solvers across different ranges of $\Card{\set}$.
    The column labeled $\sum$ indicates the number of instances within each range.}
    \label{table:performancecomparisonwithsetsize}
\end{table}

\begin{table}[h]
    \centering
    \begin{tabular}{m{7em} m{3em} m{3em} m{3em} m{3em} m{3em} m{3em} m{6em}} 
    \toprule
    & $\sum$ & SAT & ILP & MtMiner & SHD & Hitman & \toolname\\
    \midrule
    $10^1 \leq \dis <10^2$ & 1788 & 1447 & 590 & 718 & 582 & 1505 & \textbf{1563}\\
    \midrule
    $10^2 \leq \dis <10^3$ & 1145 & 278 & 57 & 279 & 46 & 325 & \textbf{372}\\
    \midrule
    $10^3 \leq \dis <10^4$ & 76 & 10 & 0 & 2 & 1 & \textbf{14} & 13\\
    \midrule
    $10^4 \leq \dis <10^5$ & 5 & 0 & 0 & 3 & \textbf{4} & 2 & 0\\
    \bottomrule
    \end{tabular}
    \caption{The performance comparison of \toolname~vis-a-vis existing hitting set solvers across different ranges of the average disjunction size. The notation $\dis$ denotes the average size of disjunctions. The column labeled $\sum$ indicates the number of instances within each range.}
    \label{table:performancecomparisonwithdisjsize}
\end{table}
We compare the number of benchmark instances solved by various MHS solvers, with the results summarized in~\Cref{table:hittingsetresult}. An instance is considered solved by a tool if it enumerates all minimal hitting sets within the specified time and memory limits.
As shown in the table, \toolname\ outperforms existing hitting set solvers by a significant margin. The runtime performance comparison between \toolname\ and other MHS solvers is illustrated in~\Cref{fig:runtime}. A point $(x, y)$ on these plots indicates that the tool solved $x$ instances, each within $y$ seconds.

To demonstrate the effectiveness of \toolname, we constructed virtual best solvers (VBS) combining \toolname~and other baseline methods. 
Specifically, we define two VBS configurations: (i) VBS$1$, which includes all baselines excluding \toolname~and (ii) VBS$2$, which includes all baselines including \toolname. 
The performance comparison is presented in~\Cref{table:virtualbestresult}.  
Notably, VBS2 solved $94$ instances that were not solved by any of the other baselines alone, highlighting the unique contributions of \toolname. 
A deeper analysis shows that SAT, ILP, MtMiner, SHD, and Hitman respectively solved $7, 0, 41, 9,$ and $42$ instances not solved by \toolname. In contrast, \toolname~solved $220, 1301, 987, 1234,$ and $123$ instances that were not solved by SAT, ILP, MtMiner, SHD, and Hitman, respectively.

Additionally, we compare the performance of minimal hitting set solvers based on three factors: the number of minimal hitting sets enumerated, the size of the instances, and the average size of the {\em disjunction}. These comparisons are visualized in~\Cref{fig:countminimalhittingset,fig:sizeminimalhittingset,fig:disjunctionminimalhittingset}.
The instance size is defined as the number of sets in $\set$, denoted by $\Card{\set}$.
The average disjunction size, denoted by $\dis$, is defined as $\frac{\sum_{S_i \in \set} \Card{S_i}}{\Card{\set}}$, representing the mean number of elements in each set $S_i$.
A point $(x, y)$ on these plots indicates that a solver successfully handled $x$ instances, where each instance has at most $y$ minimal hitting sets~(in \Cref{fig:countminimalhittingset}), has a size ($\Card{\set}$) of at most $y$~(in \Cref{fig:sizeminimalhittingset}) or features an average disjunction size of at most $y$~(in \Cref{fig:disjunctionminimalhittingset}). 
The plots show that \toolname\ is competitive with other minimal hitting set solvers in terms of the number of MHSes enumerated, the size of the instances, and the average disjunction size it can handle.

We compare the performance of \toolname~against other MHS solvers across varying ranges of instance size ($\Card{\set}$) and average disjunction size ($\dis$).
The results are shown in~\Cref{table:performancecomparisonwithsetsize,table:performancecomparisonwithdisjsize}, where $\dis$ denotes the average disjunction size. 
Both tables show that \toolname~struggles with instances that have larger $\Card{\set}$ or higher $\dis$; the number of solved instances decreases as $\Card{\set}$ or $\dis$ increases.
This performance drop may stem from the fact that \toolname~reduces MHS enumeration to a harder problem than the other solvers do.

%% file: section/conclusion.tex
\section{Conclusion}
\label{section:conclusion}
In this paper, we introduce \toolname, an ASP-based tool for minimal hitting set enumeration. 
Our ASP encoding is simple yet expressive, supporting versatile features such as enumerating minimal hitting sets of a {\em specified size, containing a given set of items, and optimal hitting sets with respect to objective functions}.
Our tool also supports advanced reasoning tasks over hitting sets, including {\em brave}, {\em cautious} reasoning, and counting~\cite{Kabir2024,KCM2024,KESHFM2022,KM2023,KTPM2025}.
By leveraging a reduction to ASP, \toolname\ takes advantage of advancements in ASP solving. 
Our experimental evaluation demonstrates that \toolname\ outperforms existing minimal hitting set solvers in enumerating all MHSes.

%% file: example.bbl
\begin{thebibliography}{10}
\providecommand{\bibitemdeclare}[2]{}
\providecommand{\surnamestart}{}
\providecommand{\surnameend}{}
\providecommand{\urlprefix}{Available at }
\providecommand{\url}[1]{\texttt{#1}}
\providecommand{\href}[2]{\texttt{#2}}
\providecommand{\urlalt}[2]{\href{#1}{#2}}
\providecommand{\doi}[1]{doi:\urlalt{https://doi.org/#1}{#1}}
\providecommand{\eprint}[1]{arXiv:\urlalt{https://arxiv.org/abs/#1}{#1}}
\providecommand{\bibinfo}[2]{#2}

\bibitemdeclare{inproceedings}{AG2009}
\bibitem{AG2009}
\bibinfo{author}{Rui \surnamestart Abreu\surnameend} \&
  \bibinfo{author}{Arjan~JC \surnamestart Van~Gemund\surnameend}
  (\bibinfo{year}{2009}): \emph{\bibinfo{title}{A Low-Cost Approximate Minimal
  Hitting Set Algorithm and its Application to Model-Based Diagnosis}}.
\newblock In: {\slshape \bibinfo{booktitle}{SARA}}, \bibinfo{volume}{9}, pp.
  \bibinfo{pages}{2--9}.

\bibitemdeclare{inproceedings}{AD2016}
\bibitem{AD2016}
\bibinfo{author}{Mario \surnamestart Alviano\surnameend} \&
  \bibinfo{author}{Carmine \surnamestart Dodaro\surnameend}
  (\bibinfo{year}{2016}): \emph{\bibinfo{title}{Completion of Disjunctive Logic
  Programs}}.
\newblock In: {\slshape \bibinfo{booktitle}{IJCAI}}, \bibinfo{volume}{16}, pp.
  \bibinfo{pages}{886--892}.

\bibitemdeclare{inproceedings}{ADLR2015}
\bibitem{ADLR2015}
\bibinfo{author}{Mario \surnamestart Alviano\surnameend},
  \bibinfo{author}{Carmine \surnamestart Dodaro\surnameend},
  \bibinfo{author}{Nicola \surnamestart Leone\surnameend} \&
  \bibinfo{author}{Francesco \surnamestart Ricca\surnameend}
  (\bibinfo{year}{2015}): \emph{\bibinfo{title}{Advances in WASP}}.
\newblock In: {\slshape \bibinfo{booktitle}{LPNMR}},
  \bibinfo{organization}{Springer}, pp. \bibinfo{pages}{40--54}, \doi{10.1007/978-3-319-23264-5_5}.

\bibitemdeclare{inproceedings}{BMR2003}
\bibitem{BMR2003}
\bibinfo{author}{James \surnamestart Bailey\surnameend},
  \bibinfo{author}{Thomas \surnamestart Manoukian\surnameend} \&
  \bibinfo{author}{Kotagiri \surnamestart Ramamohanarao\surnameend}
  (\bibinfo{year}{2003}): \emph{\bibinfo{title}{A fast algorithm for computing
  hypergraph transversals and its application in mining emerging patterns}}.
\newblock In: {\slshape \bibinfo{booktitle}{ICDM}},
  \bibinfo{organization}{IEEE}, pp. \bibinfo{pages}{485--488}, \doi{10.1109/ICDM.2003.1250958}.

\bibitemdeclare{inproceedings}{BS2005}
\bibitem{BS2005}
\bibinfo{author}{James \surnamestart Bailey\surnameend} \&
  \bibinfo{author}{Peter~J \surnamestart Stuckey\surnameend}
  (\bibinfo{year}{2005}): \emph{\bibinfo{title}{Discovery of minimal
  unsatisfiable subsets of constraints using hitting set dualization}}.
\newblock In: {\slshape \bibinfo{booktitle}{PADL}},
  \bibinfo{organization}{Springer}, pp. \bibinfo{pages}{174--186}, \doi{10.1007/978-3-540-30557-6_14}.

\bibitemdeclare{article}{BD1994}
\bibitem{BD1994}
\bibinfo{author}{Rachel \surnamestart Ben-Eliyahu\surnameend} \&
  \bibinfo{author}{Rina \surnamestart Dechter\surnameend}
  (\bibinfo{year}{1994}): \emph{\bibinfo{title}{Propositional semantics for
  disjunctive logic programs}}.
\newblock {\slshape \bibinfo{journal}{Annals of Mathematics and Artificial
  intelligence}} \bibinfo{volume}{12}, pp. \bibinfo{pages}{53--87}, \doi{10.1007/BF01530761}.


\bibitemdeclare{book}{Berge1984}
\bibitem{Berge1984}
\bibinfo{author}{Claude \surnamestart Berge\surnameend} (\bibinfo{year}{1984}):
  \emph{\bibinfo{title}{Hypergraphs: combinatorics of finite sets}}.
\newblock \bibinfo{volume}{45}, \bibinfo{publisher}{Elsevier}.

\bibitemdeclare{article}{VS2020}
\bibitem{VS2020}
\bibinfo{author}{Ren{\'e} \surnamestart van Bevern\surnameend} \&
  \bibinfo{author}{Pavel~V \surnamestart Smirnov\surnameend}
  (\bibinfo{year}{2020}): \emph{\bibinfo{title}{Optimal-size problem kernels
  for d-hitting set in linear time and space}}.
\newblock {\slshape \bibinfo{journal}{Information Processing Letters}}
  \bibinfo{volume}{163}, p. \bibinfo{pages}{105998}, \doi{10.1016/j.ipl.2020.105998}.

\bibitemdeclare{article}{BBFNPS2020}
\bibitem{BBFNPS2020}
\bibinfo{author}{Johann \surnamestart Birnick\surnameend},
  \bibinfo{author}{Thomas \surnamestart Bl{\"a}sius\surnameend},
  \bibinfo{author}{Tobias \surnamestart Friedrich\surnameend},
  \bibinfo{author}{Felix \surnamestart Naumann\surnameend},
  \bibinfo{author}{Thorsten \surnamestart Papenbrock\surnameend} \&
  \bibinfo{author}{Martin \surnamestart Schirneck\surnameend}
  (\bibinfo{year}{2020}): \emph{\bibinfo{title}{Hitting set enumeration with
  partial information for unique column combination discovery}}.
\newblock {\slshape \bibinfo{journal}{Proceedings of the VLDB Endowment}}
  \bibinfo{volume}{13}(\bibinfo{number}{12}), pp. \bibinfo{pages}{2270--2283}.

\bibitemdeclare{inproceedings}{BFSW2022}
\bibitem{BFSW2022}
\bibinfo{author}{Thomas \surnamestart Bl{\"a}sius\surnameend},
  \bibinfo{author}{Tobias \surnamestart Friedrich\surnameend},
  \bibinfo{author}{David \surnamestart Stangl\surnameend} \&
  \bibinfo{author}{Christopher \surnamestart Weyand\surnameend}
  (\bibinfo{year}{2022}): \emph{\bibinfo{title}{An Efficient Branch-and-Bound
  Solver for Hitting Set}}.
\newblock In: {\slshape \bibinfo{booktitle}{ALENEX}},
  \bibinfo{organization}{SIAM}, pp. \bibinfo{pages}{209--220}, \doi{10.1137/1.9781611977042.17}.

\bibitemdeclare{inproceedings}{BEGK2003}
\bibitem{BEGK2003}
\bibinfo{author}{Endre \surnamestart Boros\surnameend},
  \bibinfo{author}{K~\surnamestart Elbassioni\surnameend},
  \bibinfo{author}{Vladimir \surnamestart Gurvich\surnameend} \&
  \bibinfo{author}{Leonid \surnamestart Khachiyan\surnameend}
  (\bibinfo{year}{2003}): \emph{\bibinfo{title}{An efficient implementation of
  a quasi-polynomial algorithm for generating hypergraph transversals}}.
\newblock In: {\slshape \bibinfo{booktitle}{ESA}},
  \bibinfo{organization}{Springer}, pp. \bibinfo{pages}{556--567}, \doi{10.1007/978-3-540-39658-1_51}.

\bibitemdeclare{article}{BEGKM2002}
\bibitem{BEGKM2002}
\bibinfo{author}{Endre \surnamestart Boros\surnameend}, \bibinfo{author}{Khaled
  \surnamestart Elbassioni\surnameend}, \bibinfo{author}{Vladimir \surnamestart
  Gurvich\surnameend}, \bibinfo{author}{Leonid \surnamestart
  Khachiyan\surnameend} \& \bibinfo{author}{Kazuhisa \surnamestart
  Makino\surnameend} (\bibinfo{year}{2002}): \emph{\bibinfo{title}{Dual-bounded
  generating problems: All minimal integer solutions for a monotone system of
  linear inequalities}}.
\newblock {\slshape \bibinfo{journal}{SIAM Journal on Computing}}
  \bibinfo{volume}{31}(\bibinfo{number}{5}), pp. \bibinfo{pages}{1624--1643}, \doi{10.1137/S0097539701388768}.

\bibitemdeclare{article}{BGKM2003}
\bibitem{BGKM2003}
\bibinfo{author}{Endre \surnamestart Boros\surnameend},
  \bibinfo{author}{Vladimir \surnamestart Gurvich\surnameend},
  \bibinfo{author}{Leonid \surnamestart Khachiyan\surnameend} \&
  \bibinfo{author}{Kazuhisa \surnamestart Makino\surnameend}
  (\bibinfo{year}{2003}): \emph{\bibinfo{title}{On maximal frequent and minimal
  infrequent sets in binary matrices}}.
\newblock {\slshape \bibinfo{journal}{Annals of Math. and Artificial
  Intelligence}} \bibinfo{volume}{39}, pp. \bibinfo{pages}{211--221}, \doi{10.1023/A:1024605820527}.

\bibitemdeclare{article}{CFGIKK2020}
\bibitem{CFGIKK2020}
\bibinfo{author}{Francesco \surnamestart Calimeri\surnameend},
  \bibinfo{author}{Wolfgang \surnamestart Faber\surnameend},
  \bibinfo{author}{Martin \surnamestart Gebser\surnameend},
  \bibinfo{author}{Giovambattista \surnamestart Ianni\surnameend},
  \bibinfo{author}{Roland \surnamestart Kaminski\surnameend},
  \bibinfo{author}{Thomas \surnamestart Krennwallner\surnameend},
  \bibinfo{author}{Nicola \surnamestart Leone\surnameend},
  \bibinfo{author}{Marco \surnamestart Maratea\surnameend},
  \bibinfo{author}{Francesco \surnamestart Ricca\surnameend} \&
  \bibinfo{author}{Torsten \surnamestart Schaub\surnameend}
  (\bibinfo{year}{2020}): \emph{\bibinfo{title}{{ASP}-{Core}-2 input language
  format}}.
\newblock {\slshape \bibinfo{journal}{TPLP}}
  \bibinfo{volume}{20}(\bibinfo{number}{2}), pp. \bibinfo{pages}{294--309}, \doi{10.1017/S1471068419000450}.

\bibitemdeclare{inproceedings}{CKMV2011}
\bibitem{CKMV2011}
\bibinfo{author}{Karthekeyan \surnamestart Chandrasekaran\surnameend},
  \bibinfo{author}{Richard \surnamestart Karp\surnameend},
  \bibinfo{author}{Erick \surnamestart Moreno-Centeno\surnameend} \&
  \bibinfo{author}{Santosh \surnamestart Vempala\surnameend}
  (\bibinfo{year}{2011}): \emph{\bibinfo{title}{Algorithms for implicit hitting
  set problems}}.
\newblock In: {\slshape \bibinfo{booktitle}{SODA}},
  \bibinfo{organization}{SIAM}, pp. \bibinfo{pages}{614--629}, \doi{10.1137/1.9781611973082.48}.

\bibitemdeclare{article}{clark1978}
\bibitem{clark1978}
\bibinfo{author}{Keith~L \surnamestart Clark\surnameend}
  (\bibinfo{year}{1978}): \emph{\bibinfo{title}{Negation as failure}}.
\newblock {\slshape \bibinfo{journal}{Logic and data bases}}, pp.
  \bibinfo{pages}{293--322}, \doi{10.1007/978-1-4684-3384-5_11}.

\bibitemdeclare{article}{DL2005}
\bibitem{DL2005}
\bibinfo{author}{Guozhu \surnamestart Dong\surnameend} \&
  \bibinfo{author}{Jinyan \surnamestart Li\surnameend} (\bibinfo{year}{2005}):
  \emph{\bibinfo{title}{Mining border descriptions of emerging patterns from
  dataset pairs}}.
\newblock {\slshape \bibinfo{journal}{Knowledge and Information Systems}}
  \bibinfo{volume}{8}, pp. \bibinfo{pages}{178--202}, \doi{10.1007/s10115-004-0178-1}.

\bibitemdeclare{misc}{DG2017}
\bibitem{DG2017}
\bibinfo{author}{Dheeru \surnamestart Dua\surnameend} \& \bibinfo{author}{Casey
  \surnamestart Graff\surnameend} (\bibinfo{year}{2017}):
  \emph{\bibinfo{title}{{UCI} Machine Learning Repository}}.
\newblock \urlprefix\url{http://archive.ics.uci.edu/ml}.

\bibitemdeclare{inproceedings}{ES2003}
\bibitem{ES2003}
\bibinfo{author}{Niklas \surnamestart E{\'e}n\surnameend} \&
  \bibinfo{author}{Niklas \surnamestart S{\"o}rensson\surnameend}
  (\bibinfo{year}{2003}): \emph{\bibinfo{title}{An extensible {SAT}-solver}}.
\newblock In: {\slshape \bibinfo{booktitle}{SAT}},
  \bibinfo{organization}{Springer}, pp. \bibinfo{pages}{502--518}, \doi{10.1007/978-3-540-24605-3_37}.

\bibitemdeclare{article}{EG1995}
\bibitem{EG1995}
\bibinfo{author}{Thomas \surnamestart Eiter\surnameend} \&
  \bibinfo{author}{Georg \surnamestart Gottlob\surnameend}
  (\bibinfo{year}{1995}): \emph{\bibinfo{title}{On the computational cost of
  disjunctive logic programming: Propositional case}}.
\newblock {\slshape \bibinfo{journal}{Annals of Mathematics and Artificial
  Intelligence}} \bibinfo{volume}{15}, pp. \bibinfo{pages}{289--323}, \doi{10.1007/BF01536399}.


\bibitemdeclare{article}{EGL2016}
\bibitem{EGL2016}
\bibinfo{author}{Esra \surnamestart Erdem\surnameend}, \bibinfo{author}{Michael
  \surnamestart Gelfond\surnameend} \& \bibinfo{author}{Nicola \surnamestart
  Leone\surnameend} (\bibinfo{year}{2016}): \emph{\bibinfo{title}{Applications
  of answer set programming}}.
\newblock {\slshape \bibinfo{journal}{AI Magazine}}
  \bibinfo{volume}{37}(\bibinfo{number}{3}), pp. \bibinfo{pages}{53--68}, \doi{10.1609/aimag.v37i3.2678}.


\bibitemdeclare{article}{FK1996}
\bibitem{FK1996}
\bibinfo{author}{Michael~L \surnamestart Fredman\surnameend} \&
  \bibinfo{author}{Leonid \surnamestart Khachiyan\surnameend}
  (\bibinfo{year}{1996}): \emph{\bibinfo{title}{On the complexity of
  dualization of monotone disjunctive normal forms}}.
\newblock {\slshape \bibinfo{journal}{Journal of Algorithms}}
  \bibinfo{volume}{21}(\bibinfo{number}{3}), pp. \bibinfo{pages}{618--628}, \doi{10.1006/jagm.1996.0062}.

\bibitemdeclare{inproceedings}{FNS2023}
\bibitem{FNS2023}
\bibinfo{author}{Dror \surnamestart Fried\surnameend},
  \bibinfo{author}{Alexander \surnamestart Nadel\surnameend} \&
  \bibinfo{author}{Yogev \surnamestart Shalmon\surnameend}
  (\bibinfo{year}{2023}): \emph{\bibinfo{title}{{AllSAT} for combinational
  circuits}}.
\newblock In: {\slshape \bibinfo{booktitle}{SAT}},
  \bibinfo{organization}{Schloss Dagstuhl--Leibniz-Zentrum f{\"u}r Informatik},
  pp. \bibinfo{pages}{9:1--9:18}, \doi{10.4230/LIPIcs.SAT.2023.9}.


\bibitemdeclare{inproceedings}{GKS2009}
\bibitem{GKS2009}
\bibinfo{author}{Martin \surnamestart Gebser\surnameend},
  \bibinfo{author}{Benjamin \surnamestart Kaufmann\surnameend} \&
  \bibinfo{author}{Torsten \surnamestart Schaub\surnameend}
  (\bibinfo{year}{2009}): \emph{\bibinfo{title}{Solution enumeration for
  projected Boolean search problems}}.
\newblock In: {\slshape \bibinfo{booktitle}{CPAIOR}},
  \bibinfo{organization}{Springer}, pp. \bibinfo{pages}{71--86}, \doi{10.1007/978-3-642-01929-6_7}.

\bibitemdeclare{article}{GKS2012}
\bibitem{GKS2012}
\bibinfo{author}{Martin \surnamestart Gebser\surnameend},
  \bibinfo{author}{Benjamin \surnamestart Kaufmann\surnameend} \&
  \bibinfo{author}{Torsten \surnamestart Schaub\surnameend}
  (\bibinfo{year}{2012}): \emph{\bibinfo{title}{Conflict-driven answer set
  solving: From theory to practice}}.
\newblock {\slshape \bibinfo{journal}{Artificial Intelligence}}
  \bibinfo{volume}{187}, pp. \bibinfo{pages}{52--89}, \doi{10.1016/j.artint.2012.04.001}.

\bibitemdeclare{inproceedings}{GKS2013}
\bibitem{GKS2013}
\bibinfo{author}{Martin \surnamestart Gebser\surnameend},
  \bibinfo{author}{Benjamin \surnamestart Kaufmann\surnameend} \&
  \bibinfo{author}{Torsten \surnamestart Schaub\surnameend}
  (\bibinfo{year}{2013}): \emph{\bibinfo{title}{Advanced conflict-driven
  disjunctive answer set solving}}.
\newblock In: {\slshape \bibinfo{booktitle}{IJCAI}},
  \bibinfo{organization}{AAAI Press}.

\bibitemdeclare{article}{GL1991}
\bibitem{GL1991}
\bibinfo{author}{Michael \surnamestart Gelfond\surnameend} \&
  \bibinfo{author}{Vladimir \surnamestart Lifschitz\surnameend}
  (\bibinfo{year}{1991}): \emph{\bibinfo{title}{Classical negation in logic
  programs and disjunctive databases}}.
\newblock {\slshape \bibinfo{journal}{New generation computing}}
  \bibinfo{volume}{9}, pp. \bibinfo{pages}{365--385}, \doi{10.1007/BF03037169}.

\bibitemdeclare{article}{GSW1989}
\bibitem{GSW1989}
\bibinfo{author}{Russell \surnamestart Greiner\surnameend},
  \bibinfo{author}{Barbara~A \surnamestart Smith\surnameend} \&
  \bibinfo{author}{Ralph~W \surnamestart Wilkerson\surnameend}
  (\bibinfo{year}{1989}): \emph{\bibinfo{title}{A correction to the algorithm
  in Reiter's theory of diagnosis}}.
\newblock {\slshape \bibinfo{journal}{Artificial Intelligence}}
  \bibinfo{volume}{41}(\bibinfo{number}{1}), pp. \bibinfo{pages}{79--88}, \doi{10.1016/0004-3702(89)90079-9}.

\bibitemdeclare{inproceedings}{GMKT1997}
\bibitem{GMKT1997}
\bibinfo{author}{Dimitrios \surnamestart Gunopulos\surnameend},
  \bibinfo{author}{Heikki \surnamestart Mannila\surnameend},
  \bibinfo{author}{Roni \surnamestart Khardon\surnameend} \&
  \bibinfo{author}{Hannu \surnamestart Toivonen\surnameend}
  (\bibinfo{year}{1997}): \emph{\bibinfo{title}{Data mining, hypergraph
  transversals, and machine learning}}.
\newblock In: {\slshape \bibinfo{booktitle}{PODS}}, pp.
  \bibinfo{pages}{209--216}, \doi{10.1145/263661.263684}.

\bibitemdeclare{article}{GVETCSS2013}
\bibitem{GVETCSS2013}
\bibinfo{author}{Carito \surnamestart Guziolowski\surnameend},
  \bibinfo{author}{Santiago \surnamestart Videla\surnameend},
  \bibinfo{author}{Federica \surnamestart Eduati\surnameend},
  \bibinfo{author}{Sven \surnamestart Thiele\surnameend},
  \bibinfo{author}{Thomas \surnamestart Cokelaer\surnameend},
  \bibinfo{author}{Anne \surnamestart Siegel\surnameend} \&
  \bibinfo{author}{Julio \surnamestart Saez-Rodriguez\surnameend}
  (\bibinfo{year}{2013}): \emph{\bibinfo{title}{Exhaustively characterizing
  feasible logic models of a signaling network using answer set programming}}.
\newblock {\slshape \bibinfo{journal}{Bioinformatics}}
  \bibinfo{volume}{29}(\bibinfo{number}{18}), pp. \bibinfo{pages}{2320--2326}, \doi{10.1093/bioinformatics/btt393}.

\bibitemdeclare{article}{HK2011}
\bibitem{HK2011}
\bibinfo{author}{Oliver \surnamestart H{\"a}dicke\surnameend} \&
  \bibinfo{author}{Steffen \surnamestart Klamt\surnameend}
  (\bibinfo{year}{2011}): \emph{\bibinfo{title}{Computing complex metabolic
  intervention strategies using constrained minimal cut sets}}.
\newblock {\slshape \bibinfo{journal}{Metabolic engineering}}
  \bibinfo{volume}{13}(\bibinfo{number}{2}), pp. \bibinfo{pages}{204--213}, \doi{10.1016/j.ymben.2010.12.004}.

\bibitemdeclare{article}{HBC2007}
\bibitem{HBC2007}
\bibinfo{author}{C{\'e}line \surnamestart H{\'e}bert\surnameend},
  \bibinfo{author}{Alain \surnamestart Bretto\surnameend} \&
  \bibinfo{author}{Bruno \surnamestart Cr{\'e}milleux\surnameend}
  (\bibinfo{year}{2007}): \emph{\bibinfo{title}{A data mining formalization to
  improve hypergraph minimal transversal computation}}.
\newblock {\slshape \bibinfo{journal}{Fundamenta Informaticae}}
  \bibinfo{volume}{80}(\bibinfo{number}{4}), pp. \bibinfo{pages}{415--433}.

\bibitemdeclare{inproceedings}{ILSM2021}
\bibitem{ILSM2021}
\bibinfo{author}{Alexey \surnamestart Ignatiev\surnameend},
  \bibinfo{author}{Edward \surnamestart Lam\surnameend},
  \bibinfo{author}{Peter~J \surnamestart Stuckey\surnameend} \&
  \bibinfo{author}{Joao \surnamestart Marques-Silva\surnameend}
  (\bibinfo{year}{2021}): \emph{\bibinfo{title}{A scalable two stage approach
  to computing optimal decision sets}}.
\newblock In: {\slshape \bibinfo{booktitle}{AAAI}}, \bibinfo{volume}{35}, pp.
  \bibinfo{pages}{3806--3814}, \doi{10.1609/aaai.v35i5.16498}.

\bibitemdeclare{inproceedings}{AAJ2018}
\bibitem{AAJ2018}
\bibinfo{author}{Alexey \surnamestart Ignatiev\surnameend},
  \bibinfo{author}{Antonio \surnamestart Morgado\surnameend} \&
  \bibinfo{author}{Joao \surnamestart Marques{-}Silva\surnameend}
  (\bibinfo{year}{2018}): \emph{\bibinfo{title}{{PySAT:} {A} {Python} Toolkit
  for Prototyping with {SAT} Oracles}}.
\newblock In: {\slshape \bibinfo{booktitle}{SAT}}, pp.
  \bibinfo{pages}{428--437}, \doi{10.1007/978-3-319-94144-8_26}.


\bibitemdeclare{article}{Kabir2024}
\bibitem{Kabir2024}
\bibinfo{author}{Mohimenul \surnamestart Kabir\surnameend}
  (\bibinfo{year}{2024}): \emph{\bibinfo{title}{Minimal model counting via
  knowledge compilation}}.
\newblock {\slshape \bibinfo{journal}{arXiv preprint arXiv:2409.10170}}, \doi{10.48550/arXiv.2409.10170}.

\bibitemdeclare{inproceedings}{KCM2024}
\bibitem{KCM2024}
\bibinfo{author}{Mohimenul \surnamestart Kabir\surnameend},
  \bibinfo{author}{Supratik \surnamestart Chakraborty\surnameend} \&
  \bibinfo{author}{Kuldeep~S \surnamestart Meel\surnameend}
  (\bibinfo{year}{2024}): \emph{\bibinfo{title}{Exact ASP counting with compact
  encodings}}.
\newblock In: {\slshape \bibinfo{booktitle}{AAAI}}, \bibinfo{volume}{38}, pp.
  \bibinfo{pages}{10571--10580}, \doi{10.1609/aaai.v38i9.28927}.

\bibitemdeclare{inproceedings}{KESHFM2022}
\bibitem{KESHFM2022}
\bibinfo{author}{Mohimenul \surnamestart Kabir\surnameend},
  \bibinfo{author}{Flavio~O \surnamestart Everardo\surnameend},
  \bibinfo{author}{Ankit~K \surnamestart Shukla\surnameend},
  \bibinfo{author}{Markus \surnamestart Hecher\surnameend},
  \bibinfo{author}{Johannes~Klaus \surnamestart Fichte\surnameend} \&
  \bibinfo{author}{Kuldeep~S \surnamestart Meel\surnameend}
  (\bibinfo{year}{2022}): \emph{\bibinfo{title}{ApproxASP--a scalable
  approximate answer set counter}}.
\newblock In: {\slshape \bibinfo{booktitle}{AAAI}}, \bibinfo{volume}{36}, pp.
  \bibinfo{pages}{5755--5764}, \doi{10.1609/aaai.v36i5.20518}.

\bibitemdeclare{inproceedings}{KM2023}
\bibitem{KM2023}
\bibinfo{author}{Mohimenul \surnamestart Kabir\surnameend} \&
  \bibinfo{author}{Kuldeep~S \surnamestart Meel\surnameend}
  (\bibinfo{year}{2023}): \emph{\bibinfo{title}{A Fast and Accurate ASP
  Counting Based Network Reliability Estimator.}}
\newblock In: {\slshape \bibinfo{booktitle}{LPAR}}, \bibinfo{volume}{94}, pp.
  \bibinfo{pages}{270--287}, \doi{10.29007/kc6q}.

\bibitemdeclare{article}{KM2024}
\bibitem{KM2024}
\bibinfo{author}{Mohimenul \surnamestart Kabir\surnameend} \&
  \bibinfo{author}{Kuldeep~S \surnamestart Meel\surnameend}
  (\bibinfo{year}{2024}): \emph{\bibinfo{title}{On Lower Bounding Minimal Model
  Count}}.
\newblock {\slshape \bibinfo{journal}{TPLP}}
  \bibinfo{volume}{24}(\bibinfo{number}{4}), pp. \bibinfo{pages}{586--605},
  \doi{10.1017/S147106842400036X}.

\bibitemdeclare{article}{KM2025}
\bibitem{KM2025}
\bibinfo{author}{Mohimenul \surnamestart Kabir\surnameend} \&
  \bibinfo{author}{Kuldeep~S \surnamestart Meel\surnameend}
  (\bibinfo{year}{2025}): \emph{\bibinfo{title}{An ASP-Based Framework for
  MUSes}}.
\newblock {\slshape \bibinfo{journal}{arXiv preprint arXiv:2507.03929 (to
  appear ICLP 2025 Technical Communications)}}, \doi{10.48550/arXiv.2507.03929}.

\bibitemdeclare{article}{KTPM2025}
\bibitem{KTPM2025}
\bibinfo{author}{Mohimenul \surnamestart Kabir\surnameend},
  \bibinfo{author}{Van-Giang \surnamestart Trinh\surnameend},
  \bibinfo{author}{Samuel \surnamestart Pastva\surnameend} \&
  \bibinfo{author}{Kuldeep~S \surnamestart Meel\surnameend}
  (\bibinfo{year}{2025}): \emph{\bibinfo{title}{Scalable Counting of Minimal
  Trap Spaces and Fixed Points in Boolean Networks}}.
\newblock {\slshape \bibinfo{journal}{arXiv preprint arXiv:2506.06013 (to
  appear CP 2025)}}, \doi{10.48550/arXiv.2506.06013}.

\bibitemdeclare{incollection}{Karp2009}
\bibitem{Karp2009}
\bibinfo{author}{Richard~M \surnamestart Karp\surnameend}
  (\bibinfo{year}{2009}): \emph{\bibinfo{title}{Reducibility among
  combinatorial problems}}.
\newblock In: {\slshape \bibinfo{booktitle}{50 Years of Integer Programming
  1958-2008: from the Early Years to the State-of-the-Art}},
  \bibinfo{publisher}{Springer}, pp. \bibinfo{pages}{219--241}, \doi{10.1007/978-3-540-68279-0_8}.

\bibitemdeclare{article}{KS2005}
\bibitem{KS2005}
\bibinfo{author}{Dimitris \surnamestart Kavvadias\surnameend} \&
  \bibinfo{author}{Elias \surnamestart Stavropoulos\surnameend}
  (\bibinfo{year}{2005}): \emph{\bibinfo{title}{An efficient algorithm for the
  transversal hypergraph generation}}.
\newblock {\slshape \bibinfo{journal}{Journal of Graph Algorithms and
  Applications}} \bibinfo{volume}{9}(\bibinfo{number}{2}), pp.
  \bibinfo{pages}{239--264}, \doi{10.7155/jgaa.00107}.

\bibitemdeclare{article}{KG2004}
\bibitem{KG2004}
\bibinfo{author}{Steffen \surnamestart Klamt\surnameend} \&
  \bibinfo{author}{Ernst~Dieter \surnamestart Gilles\surnameend}
  (\bibinfo{year}{2004}): \emph{\bibinfo{title}{Minimal cut sets in biochemical
  reaction networks}}.
\newblock {\slshape \bibinfo{journal}{Bioinformatics}}
  \bibinfo{volume}{20}(\bibinfo{number}{2}), pp. \bibinfo{pages}{226--234}, \doi{10.1093/bioinformatics/btg395}.

\bibitemdeclare{article}{KNPetal2009}
\bibitem{KNPetal2009}
\bibinfo{author}{Tolga \surnamestart Kurtoglu\surnameend},
  \bibinfo{author}{Sriram \surnamestart Narasimhan\surnameend},
  \bibinfo{author}{Scott \surnamestart Poll\surnameend}, \bibinfo{author}{David
  \surnamestart Garcia\surnameend}, \bibinfo{author}{Lukas \surnamestart
  Kuhn\surnameend}, \bibinfo{author}{Johan \surnamestart de~Kleer\surnameend},
  \bibinfo{author}{Arjan \surnamestart van Gemund\surnameend} \&
  \bibinfo{author}{Alexander \surnamestart Feldman\surnameend}
  (\bibinfo{year}{2009}): \emph{\bibinfo{title}{First international diagnosis
  competition-DXC’09}}.
\newblock {\slshape \bibinfo{journal}{DX}} \bibinfo{volume}{9}, pp.
  \bibinfo{pages}{383--396}.

\bibitemdeclare{inproceedings}{LSK2023}
\bibitem{LSK2023}
\bibinfo{author}{Anna~L.D. \surnamestart Latour\surnameend},
  \bibinfo{author}{Arunabha \surnamestart Sen\surnameend} \&
  \bibinfo{author}{Kuldeep~S. \surnamestart Meel\surnameend}
  (\bibinfo{year}{2023}): \emph{\bibinfo{title}{Solving the Identifying Code
  Set Problem with Grouped Independent Support}}.
\newblock In: {\slshape \bibinfo{booktitle}{IJCAI}}, pp.
  \bibinfo{pages}{1971--1978}, \doi{10.24963/ijcai.2023/219}.

\bibitemdeclare{inproceedings}{LL2003}
\bibitem{LL2003}
\bibinfo{author}{Joohyung \surnamestart Lee\surnameend} \&
  \bibinfo{author}{Vladimir \surnamestart Lifschitz\surnameend}
  (\bibinfo{year}{2003}): \emph{\bibinfo{title}{Loop formulas for disjunctive
  logic programs}}.
\newblock In: {\slshape \bibinfo{booktitle}{ICLP}},
  \bibinfo{organization}{Springer}, pp. \bibinfo{pages}{451--465},
  \doi{10.1007/978-3-540-24599-5_31}.

\bibitemdeclare{inproceedings}{LMLX2010}
\bibitem{LMLX2010}
\bibinfo{author}{Charles~E \surnamestart Leiserson\surnameend},
  \bibinfo{author}{Marc \surnamestart Moreno~Maza\surnameend},
  \bibinfo{author}{Liyun \surnamestart Li\surnameend} \&
  \bibinfo{author}{Yuzhen \surnamestart Xie\surnameend} (\bibinfo{year}{2010}):
  \emph{\bibinfo{title}{Parallel computation of the minimal elements of a
  poset}}.
\newblock In: {\slshape \bibinfo{booktitle}{International Workshop on Parallel
  and Symbolic Computation}}, pp. \bibinfo{pages}{53--62}, \doi{10.1145/1837210.1837221}.

\bibitemdeclare{article}{LY2002}
\bibitem{LY2002}
\bibinfo{author}{Lin \surnamestart Li\surnameend} \& \bibinfo{author}{Jiang
  \surnamestart Yunfei\surnameend} (\bibinfo{year}{2002}):
  \emph{\bibinfo{title}{Computing minimal hitting sets with genetic
  algorithm}}.
\newblock {\slshape \bibinfo{journal}{Algorithmica}}
  \bibinfo{volume}{32}(\bibinfo{number}{1}), pp. \bibinfo{pages}{95--106}.

\bibitemdeclare{article}{LYZR2021}
\bibitem{LYZR2021}
\bibinfo{author}{Zhang \surnamestart Li\surnameend}, \bibinfo{author}{Wang
  \surnamestart Yisong\surnameend}, \bibinfo{author}{Xie \surnamestart
  Zhongtao\surnameend} \& \bibinfo{author}{Feng \surnamestart
  Renyan\surnameend} (\bibinfo{year}{2021}): \emph{\bibinfo{title}{Computing
  Propositional Minimal Models: {MiniSAT}-Based Approaches}}.
\newblock {\slshape \bibinfo{journal}{Journal of Computer Research and
  Development}} \bibinfo{volume}{58}(\bibinfo{number}{11}), pp.
  \bibinfo{pages}{2515--2523}.

\bibitemdeclare{inproceedings}{Lierler2005}
\bibitem{Lierler2005}
\bibinfo{author}{Yuliya \surnamestart Lierler\surnameend}
  (\bibinfo{year}{2005}): \emph{\bibinfo{title}{Cmodels--{SAT}-based
  disjunctive answer set solver}}.
\newblock In: {\slshape \bibinfo{booktitle}{LPNMR}},
  \bibinfo{organization}{Springer}, pp. \bibinfo{pages}{447--451}, \doi{10.1007/11546207_44}.

\bibitemdeclare{incollection}{MT1999}
\bibitem{MT1999}
\bibinfo{author}{Victor~W \surnamestart Marek\surnameend} \&
  \bibinfo{author}{Miroslaw \surnamestart Truszczy{\'n}ski\surnameend}
  (\bibinfo{year}{1999}): \emph{\bibinfo{title}{Stable models and an
  alternative logic programming paradigm}}.
\newblock In: {\slshape \bibinfo{booktitle}{The Logic Programming Paradigm}},
  \bibinfo{publisher}{Springer}, pp. \bibinfo{pages}{375--398}.

\bibitemdeclare{article}{MT1991}
\bibitem{MT1991}
\bibinfo{author}{Wiktor \surnamestart Marek\surnameend} \&
  \bibinfo{author}{Miros{\l}aw \surnamestart Truszczy{\'n}ski\surnameend}
  (\bibinfo{year}{1991}): \emph{\bibinfo{title}{Autoepistemic logic}}.
\newblock {\slshape \bibinfo{journal}{Journal of the ACM (JACM)}}
  \bibinfo{volume}{38}(\bibinfo{number}{3}), pp. \bibinfo{pages}{587--618}, \doi{10.1145/116825.116836}.

\bibitemdeclare{inproceedings}{MSS2023}
\bibitem{MSS2023}
\bibinfo{author}{Gabriele \surnamestart Masina\surnameend},
  \bibinfo{author}{Giuseppe \surnamestart Spallitta\surnameend} \&
  \bibinfo{author}{Roberto \surnamestart Sebastiani\surnameend}
  (\bibinfo{year}{2023}): \emph{\bibinfo{title}{On {CNF} Conversion for
  Disjoint {SAT} Enumeration}}.
\newblock In: {\slshape \bibinfo{booktitle}{{SAT}}}, pp.
  \bibinfo{pages}{15:1--15:16}, \doi{10.4230/LIPIcs.SAT.2023.15}.

\bibitemdeclare{article}{MTC2014}
\bibitem{MTC2014}
\bibinfo{author}{Gary \surnamestart McGuire\surnameend},
  \bibinfo{author}{Bastian \surnamestart Tugemann\surnameend} \&
  \bibinfo{author}{Gilles \surnamestart Civario\surnameend}
  (\bibinfo{year}{2014}): \emph{\bibinfo{title}{There is no 16-clue Sudoku:
  Solving the Sudoku minimum number of clues problem via hitting set
  enumeration}}.
\newblock {\slshape \bibinfo{journal}{Experimental Mathematics}}
  \bibinfo{volume}{23}(\bibinfo{number}{2}), pp. \bibinfo{pages}{190--217}, \doi{10.1080/10586458.2013.870056}.

\bibitemdeclare{inproceedings}{MPM2015}
\bibitem{MPM2015}
\bibinfo{author}{Carlos \surnamestart Menc{\i}a\surnameend},
  \bibinfo{author}{Alessandro \surnamestart Previti\surnameend} \&
  \bibinfo{author}{Joao \surnamestart Marques-Silva\surnameend}
  (\bibinfo{year}{2015}): \emph{\bibinfo{title}{Literal-based MCS extraction}}.
\newblock In: {\slshape \bibinfo{booktitle}{IJCAI}}, \bibinfo{volume}{15}, pp.
  \bibinfo{pages}{1973--1979}.

\bibitemdeclare{inproceedings}{MU2013}
\bibitem{MU2013}
\bibinfo{author}{Keisuke \surnamestart Murakami\surnameend} \&
  \bibinfo{author}{Takeaki \surnamestart Uno\surnameend}
  (\bibinfo{year}{2013}): \emph{\bibinfo{title}{Efficient algorithms for
  dualizing large-scale hypergraphs}}.
\newblock In: {\slshape \bibinfo{booktitle}{ALENEX}},
  \bibinfo{organization}{SIAM}, pp. \bibinfo{pages}{1--13}, \doi{10.1137/1.9781611972931.1}.

\bibitemdeclare{article}{RNW1983}
\bibitem{RNW1983}
\bibinfo{author}{James~A \surnamestart Reggia\surnameend},
  \bibinfo{author}{Dana~S \surnamestart Nau\surnameend} \&
  \bibinfo{author}{Pearl~Y \surnamestart Wang\surnameend}
  (\bibinfo{year}{1983}): \emph{\bibinfo{title}{Diagnostic expert systems based
  on a set covering model}}.
\newblock {\slshape \bibinfo{journal}{International Journal of Man-Machine
  Studies}} \bibinfo{volume}{19}(\bibinfo{number}{5}), pp.
  \bibinfo{pages}{437--460}, \doi{10.1016/S0020-7373(83)80065-0}.

\bibitemdeclare{article}{Reiter1987}
\bibitem{Reiter1987}
\bibinfo{author}{Raymond \surnamestart Reiter\surnameend}
  (\bibinfo{year}{1987}): \emph{\bibinfo{title}{A theory of diagnosis from
  first principles}}.
\newblock {\slshape \bibinfo{journal}{Artificial intelligence}}
  \bibinfo{volume}{32}(\bibinfo{number}{1}), pp. \bibinfo{pages}{57--95}, \doi{10.1016/0004-3702(87)90062-2}.


\bibitemdeclare{inproceedings}{SU2006}
\bibitem{SU2006}
\bibinfo{author}{Ken \surnamestart Satoh\surnameend} \&
  \bibinfo{author}{Takeaki \surnamestart Uno\surnameend}
  (\bibinfo{year}{2006}): \emph{\bibinfo{title}{Enumerating minimal
  explanations by minimal hitting set computation}}.
\newblock In: {\slshape \bibinfo{booktitle}{KSEM}},
  \bibinfo{organization}{Springer}, pp. \bibinfo{pages}{354--365}, \doi{10.1007/11811220_30}.

\bibitemdeclare{inproceedings}{Toda2013}
\bibitem{Toda2013}
\bibinfo{author}{Takahisa \surnamestart Toda\surnameend}
  (\bibinfo{year}{2013}): \emph{\bibinfo{title}{Hypergraph transversal
  computation with binary decision diagrams}}.
\newblock In: {\slshape \bibinfo{booktitle}{SEA}},
  \bibinfo{organization}{Springer}, pp. \bibinfo{pages}{91--102}, \doi{10.1007/978-3-642-38527-8_10}.

\bibitemdeclare{article}{TWS2009}
\bibitem{TWS2009}
\bibinfo{author}{Cong~T \surnamestart Trinh\surnameend}, \bibinfo{author}{Aaron
  \surnamestart Wlaschin\surnameend} \& \bibinfo{author}{Friedrich
  \surnamestart Srienc\surnameend} (\bibinfo{year}{2009}):
  \emph{\bibinfo{title}{Elementary mode analysis: a useful metabolic pathway
  analysis tool for characterizing cellular metabolism}}.
\newblock {\slshape \bibinfo{journal}{Applied Microbiology and Biotechnology}}
  \bibinfo{volume}{81}, pp. \bibinfo{pages}{813--826}, \doi{10.1007/s00253-008-1770-1}.

\bibitemdeclare{article}{Truszczynski2011}
\bibitem{Truszczynski2011}
\bibinfo{author}{Miros{\l}aw \surnamestart Truszczy{\'n}ski\surnameend}
  (\bibinfo{year}{2011}): \emph{\bibinfo{title}{Trichotomy and dichotomy
  results on the complexity of reasoning with disjunctive logic programs}}.
\newblock {\slshape \bibinfo{journal}{TPLP}}
  \bibinfo{volume}{11}(\bibinfo{number}{6}), pp. \bibinfo{pages}{881--904}, \doi{10.1017/S1471068410000463}.

\bibitemdeclare{article}{VBBZ2013}
\bibitem{VBBZ2013}
\bibinfo{author}{Paola \surnamestart Vera-Licona\surnameend},
  \bibinfo{author}{Eric \surnamestart Bonnet\surnameend},
  \bibinfo{author}{Emmanuel \surnamestart Barillot\surnameend} \&
  \bibinfo{author}{Andrei \surnamestart Zinovyev\surnameend}
  (\bibinfo{year}{2013}): \emph{\bibinfo{title}{{OCSANA}: optimal combinations
  of interventions from network analysis}}.
\newblock {\slshape \bibinfo{journal}{Bioinformatics}}
  \bibinfo{volume}{29}(\bibinfo{number}{12}), pp. \bibinfo{pages}{1571--1573}, \doi{10.1093/bioinformatics/btt195}.

\bibitemdeclare{article}{VO2000}
\bibitem{VO2000}
\bibinfo{author}{Staal \surnamestart Vinterbo\surnameend} \&
  \bibinfo{author}{Aleksander \surnamestart {\O}hrn\surnameend}
  (\bibinfo{year}{2000}): \emph{\bibinfo{title}{Minimal approximate hitting
  sets and rule templates}}.
\newblock {\slshape \bibinfo{journal}{International Journal of Approximate
  Reasoning}} \bibinfo{volume}{25}(\bibinfo{number}{2}), pp.
  \bibinfo{pages}{123--143}, \doi{10.1016/S0888-613X(00)00051-7}.

\bibitemdeclare{article}{WK2022}
\bibitem{WK2022}
\bibinfo{author}{Franz \surnamestart Wotawa\surnameend} \&
  \bibinfo{author}{David \surnamestart Kaufmann\surnameend}
  (\bibinfo{year}{2022}): \emph{\bibinfo{title}{Model-based reasoning using
  answer set programming}}.
\newblock {\slshape \bibinfo{journal}{Applied Intelligence}}
  \bibinfo{volume}{52}(\bibinfo{number}{15}), pp.
  \bibinfo{pages}{16993--17011}, \doi{10.1007/s10489-022-03272-2}.

\bibitemdeclare{article}{ZS2005}
\bibitem{ZS2005}
\bibinfo{author}{Ionela \surnamestart Zevedei-Oancea\surnameend} \&
  \bibinfo{author}{Stefan \surnamestart Schuster\surnameend}
  (\bibinfo{year}{2005}): \emph{\bibinfo{title}{A theoretical framework for
  detecting signal transfer routes in signalling networks}}.
\newblock {\slshape \bibinfo{journal}{Computers \& Chemical Engineering}}
  \bibinfo{volume}{29}(\bibinfo{number}{3}), pp. \bibinfo{pages}{597--617}, \doi{10.1016/j.compchemeng.2004.08.026}.

\end{thebibliography}
